\documentclass[journal]{IEEEtran}
\ifCLASSINFOpdf
\else
\fi
\usepackage{subfigure} 
\usepackage{graphicx}
\usepackage{amssymb}
\usepackage{amsmath}
\usepackage{epstopdf}
\usepackage{cite}
\usepackage{color}
\newtheorem{theorem}{Theorem}

\newtheorem{lemma}{Lemma}
\newtheorem{definition}{Definition}
\newtheorem{assumption}{Assumption}
\newtheorem{remark}{Remark}
\newtheorem{proof}{Proof}
\usepackage{tikz,xcolor,hyperref}
\definecolor{lime}{HTML}{A6CE39}
\DeclareRobustCommand{\orcidicon}{
	\begin{tikzpicture}
		\draw[lime, fill=lime] (0,0) 
		circle [radius=0.16] 
		node[white] {{\fontfamily{qag}\selectfont \tiny ID}}; 
		\draw[white, fill=white] (-0.0625,0.095) 
		circle [radius=0.007];	  
	\end{tikzpicture}
	\hspace{-2mm}}
\foreach \x in {A, ..., Z}{
	\expandafter\xdef\csname orcid\x\endcsname{\noexpand\href{https://orcid.org/\csname orcidauthor\x\endcsname}{\noexpand\orcidicon}}
}
\makeatother

\begin{document}

\title{A PI+R Control Scheme Based on Multi-agent Systems for Economic Dispatch in Isolated BESSs}

\author{Yalin Zhang\orcidA{},
        Zhongxin Liu\orcidB{},~\IEEEmembership{Member,~IEEE,}
        and Zengqiang Chen\orcidC{}
\thanks{Manuscript received XX, XX; revised XX, XX;
	accepted XX, XX. This work is supported by the National Natural Science Foundation of China (Grant No. 62103203) and the General Terminal IC Interdisciplinary Science Center of Nankai University (Corresponding author: Zhongxin Liu.). 
	\par The authors are with the College of Artificial Intelligence, Nankai University, Tianjin 300350, and also with the Tianjin Key Laboratory of Interventional Brain-Computer Interface and Intelligent Rehabilitation, Nankai University, Tianjin 300350, China (e-mail: zhangyl@mail.nankai.edu.cn; lzhx@nankai.edu.cn; chenzq@nankai.edu.cn).}}

\markboth{IEEE/CAA JOURNAL OF AUTOMATICA SINICA, VOL. XX, NO. XX, XX XX}%
{Zhang \MakeLowercase{\textit{et al.}}: A PI+R Control Scheme Based on Multi-agent Systems for Economic Dispatch in Isolated BESSs}

\maketitle

\begin{abstract}
	Battery energy storage systems (BESSs) are widely used in smart grids. However, power consumed by inner impedance and the capacity degradation of each battery unit become particularly severe, which has resulted in an increase in operating costs. The general economic dispatch (ED) algorithm based on marginal cost (MC) consensus is usually a proportional (P) controller, which encounters the defects of slow convergence speed and low control accuracy. In order to solve the distributed ED problem of the isolated BESS network with excellent dynamic and steady-state performance, we \textcolor{blue}{attempt} to design a proportional integral (PI) controller with a reset mechanism (PI+R) to asymptotically promote MC consensus and total power mismatch towards 0 in this paper. To be frank, the integral term in the PI controller is reset to 0 at an appropriate time when the proportional term undergoes a zero crossing, which accelerates convergence, improves control accuracy, and avoids overshoot. The eigenvalues of the system under a PI+R controller is well analyzed, ensuring the regularity of the system and enabling the reset mechanism. To ensure supply and demand balance within the isolated BESSs, a centralized reset mechanism is introduced, so that the controller is distributed in a flow set and centralized in a jump set. To cope with Zeno behavior and input delay, a dwell time that the system resides in a flow set is given. Based on this, the system with input delays can be reduced to a time-delay free system. Considering the capacity limitation of the battery, a modified MC scheme with PI+R controller is designed. The correctness of the designed scheme is verified through relevant simulations.
\end{abstract}

\begin{IEEEkeywords}
Economic dispatch, battery energy storage system, multi-agent system, reset control, distributed control.
\end{IEEEkeywords}

\IEEEpeerreviewmaketitle

\section{Introduction}
\IEEEPARstart{B}{attery} energy storage system (BESS), as a supplementary power supply, is an important component of the uninterruptible power supply \cite{yong_capacity_2023}. In addition, BESSs can smooth peak and valley periods to alleviate community electricity shortages \cite{wangApplicationEnergyStorage2022, solyaliComprehensiveStateoftheartReview2022}. BESSs has become a necessity integrated into the smart grid and energy internet. The introduction of BESSs can reduce pollutant emissions, and expenses of communities and users can be reduced due to the transfer of power over time \cite{duanDistributedAlgorithmBased2021,ghanjatiOptimalSizingEnergy2022}.
With the continuous increase of various load powers, the internal impedance power consumption of each battery cell is also increasing, which not only causes unnecessary waste, but the generated heat seriously threatens the safe operation of BESS \cite{yangModellingOptimalEnergy2022,duanDistributedAlgorithmBased2021}. In addition, the increase in power accelerates the degradation of capacity, which seriously affects the service life of battery cells, and so on. All of \textcolor{blue}{these have} led to a continuous increase in the operating costs of BESSs \cite{duanDistributedAlgorithmBased2021,rouholaminiReviewModelingManagement2022}.
\par To reduce the operating expenses of microgrids (MGs), scholars have made a lot of efforts. Previously, some excellent centralized methods were \textcolor{blue}{proposed} to address energy management problems in MGs, such as $\lambda$ iteration \cite{wangDistributedOptimalPower2021}, gradient search \cite{liNewDistributedEnergy2021,hassanImprovedMantaRay2021} and some intelligent search algorithms \cite{wangMultiagentbasedCollaborativeRegulation2022}, under which energy management problems undeniably have been well addressed. However, the centralized method requires high communication facility construction cost, heavy computing burden on the central node, and is vulnerable to single point of failure, and lacks an effective privacy protection mechanism. In view of these shortcomings, the distributed methods have been chosen by many authors to design controllers.
\par Firstly, \textcolor{blue}{we} review the distributed ED scheme for MGs. According to recent researches, most distributed solutions for ED problems in MGs are designed based on multi-agent systems (MASs) with continuous or discrete time dynamics. The \textcolor{blue}{scholars} design a distributed energy management system based on MASs in \cite{ullahComputationallyEfficientConsensusBased2021} and \cite{wangDisEHPPCEnablingHeterogeneous2022}, thus facilitating the formation of privacy protection mechanisms. In order to address the global ED problem of AC/DC interconnected microgrids, a distributed scheme is designed in \cite{liDistributedControlStrategy2021} and based on this, the bus voltage is restored. Subsequently, a distributed ED scheme with event triggering mechanism is proposed in \cite{pengDistributedPeriodicEventTriggered2022}. To achieve seamless connection between isolated and grid-connected modes, a unified ED scheme is developed in \cite{chenDistributedEconomicDispatch2021}. The above are all ED schemes with discrete time dynamics, and schemes with continuous time dynamics are also very rich. For example, the \textcolor{blue}{scholars} design a distributed ED scheme by using adaptive methods in \cite{songCostBasedAdaptiveDroop2021} to participate in the secondary control. A local event driven ED scheme is designed in \cite{sahooLocalizedEventDrivenResilient2021} to address data integrity attacks. To accelerate the convergence of the system, a novel fully distributed fixed time scheme is designed in \cite{zaeryNovelFullyDistributed2021} to solve the ED problem of multiple microgrids.
\par All the schemes designed above are proportional (P) controllers and based on the consensus of marginal cost (MC) \cite{spanglerPowerGenerationOperation2014,tanExtensionsLocationalMarginal2022}, which is effective and easy to design controllers.
Fortunately, this scheme based on MC consensus has been successfully transplanted and preliminarily \textcolor{blue}{solve} the ED problem of BESSs in \cite{yuFrequencySynchronizationPower2021b, zhaoDifferentialPrivacyEnergy2022, chenDistributedCooperativeControl2021, jinManageDistributedEnergy2022}. 
\par Although the authors in the above studies have proposed reliable solutions for the DED problem of BESSs, there is still a lot of work to be done. \textcolor{blue}{Without detailed modeling, the authors in \cite{yuFrequencySynchronizationPower2021b} directly introduce a convex function as the cost function for a BESS network.} The same issue also arises in \cite{zhaoDifferentialPrivacyEnergy2022}, where the coefficients of the cost function do not specify their composition in detail. This issue seems to be modeled as charging and discharging efficiency in \cite{chenDistributedCooperativeControl2021}, which has been further alleviated. In view of this, the authors design an ED scheme based on MC consensus in \cite{jinManageDistributedEnergy2022}. Further, P controllers are still the mainstream control scheme in the above results. This inevitably causes the system to suffer from low control accuracy and slow convergence speed. Besides, although finite/fixed time controllers can accelerate system convergence and improve control accuracy, the fractional order terms in the controller pose significant difficulties for the physical implementation of the controller. Recently, researchers have proposed a PI controller with a reset mechanism \cite{banosResetControlSystems2011, huResetControlConsensus2022,  mengResetControlSynchronization2019}, which not only has the advantages of PI controllers such as fast response speed and high control accuracy, but also eliminates some drawbacks such as large overshoot and oscillation. Moreover, we have developed an MC distributed consensus controller with reset mechanism for grid connected BESSs in recent research \cite{10302354}. Therefore, in this paper, we attempt to design a PI+R controller with continue dynamics for the isolated BESS network to improve the dynamic and steady performance, where the internal power consumption and capacity degradation of each battery cell are considered in the objective function. Our main work in this paper is listed as follows:
\begin{itemize}
	\item For isolated BESSs, a PI+R controller is designed so that MCs reach consensus and power mismatches converge to 0. The reset mechanism is designed based on zero crossing. The parameter conditions for the reset mechanism to function are given. The stability is well analyzed.
	\item The balance between supply and demand and capacity constraints are well maintained. To ensure supply and demand balance, a centralized reset mechanism has been introduced. Considering capacity constraints, the convergence rate of the modified MC scheme by time trigger is also accelerated by the reset mechanism.
	\item The performance of the controller under Zeno behavior and input delay is analyzed. The dwell time is introduced to ensure non Zeno behavior.
	Therefore, the system with input delays can be reduced to a time-delay free system.
\end{itemize}
\par Next, the multi-objective ED problem for BESSs and its KKT condition are analyzed in section II, and control objectives are given. And the theoretical analysis of the designed scheme is arranged in section III, where supply and demand balance, Zeno behavior and time delay effects are all analyzed. Several simulation cases and related analysis are arranged in section IV. Finally, the conclusion is \textcolor{blue}{drawn} in section V.
\section{Preliminaries}
For an isolated MG containing $n$ BESSs, as shown in Fig. \ref{figei1}, a certain amount of power is consumed by the internal impedance of each battery unit, coupled with capacity degradation, which requires MG operators to pay more for BESSs. In this section, the power loss and capacity degradation are analyzed and modeled to construct a management expenditure function for BESS. The KKT condition for optimal ED are well analyzed, thereby providing control objectives.
\subsection{BESS and its expenditure model}
The power consumed by the internal impedance of the battery is a linear function of power \cite{yuFrequencySynchronizationPower2021b, chenDistributedCooperativeControl2021, jinManageDistributedEnergy2022, zhaoDifferentialPrivacyEnergy2022}, i.e.,
\begin{equation}
	\label{1}
	f_{Li}=\gamma_iP_i,
\end{equation}
where the inner power loss ratio $\gamma_i$ of a BESS is often proportional to output power $P_i$ on a small time scale \cite{yuFrequencySynchronizationPower2021b, chenDistributedCooperativeControl2021, jinManageDistributedEnergy2022, zhaoDifferentialPrivacyEnergy2022}, i.e.,
\begin{equation}
	\label{2}
	\gamma_i=\pi_iP_i+\Omega_i,
\end{equation}
where $\pi_i$ and $\Omega_i$ are two constant coefficients. Battery capacity degradation is proportional to $f_{Li}$ \cite{hossainlipuReviewControllersOptimizations2022} and can be expressed as
\begin{equation}
	\label{3}
	f_{Ci}=m_if_{Li},
\end{equation}
where $m_i$ is called as the capacity degradation coefficient. Here, we construct an objective function in (\ref{4}) considering power loss and capacity degradation, which are subjected to the supply and demand balance.
\begin{equation}
	\label{4}
	\begin{array}{l}
		F_d=\sum_{i=1}^{n}{(\gamma_iP_i+f_{Ci})}\\
		\;\;\;\;\;=\sum_{i=1}^{n}{(1+m_i)(\pi_iP^2_i+\Omega_iP_i)} \\
		s.t. \sum_{i=1}^{n}{D_i}-\sum_{i=1}^{n}{P_i}=0
	\end{array},
\end{equation}
where $D_i$ is the local estimated active power load. In this paper, the power loss on the line is not considered.
\subsection{Analysis for optimal ED}
\par Obviously, the expenditure functions (\ref{4}) is a convex function of output power, which implies that the optimal condition can be obtained. In view of this, construct a Lagrangian function $L_{ac}(P_1,P_2,\cdots,P_n)$ for (\ref{4}),
\begin{equation}
	\label{5}
	L_{ac}(P_1,P_2,\cdots,P_n)=F_d+\omega(\sum_{i=1}^{n}{D_i}-\sum_{i=1}^n{P_i}),
\end{equation}
where $\omega$ is a Lagrange multiplier. Thus, the following Karush-Kuhn-Tucker (KKT) condition for BESSs can be obtained by calculating the gradient $\nabla L_{ac}(P_1,P_2,\cdots,P_n)$ as follow
\begin{equation}
	\nabla L_{ac}(P_1,P_2,\cdots,P_n)=2\frac{\partial L_{ac}(P_1,P_2,\cdots,P_n)}{\partial P_i}=\lambda_i-\omega, 
	\label{6}
\end{equation}
where $\lambda_i=(1+m_i)(2\pi_iP_i+\Omega_i)$ is so-called MC. Thus, (\ref{6}) can derive the optimal ED condition, i.e., $\lambda_i=\lambda_j$ for $\forall i,\,j\in\{1,2,\cdots,n\}$, for isolated BESSs by letting $\frac{\partial L_{ac}(P_1,P_2,\cdots,P_n)}{\partial P_i}=0$. From this, \textcolor{blue}{one} can conclude that the total expenditure for BESSs is the lowest if MCs of BESSs reach consensus. Denote $\alpha_i=(1+m_i)\Omega_i$ and $\beta_i=(1+m_i)\pi_i$. Thus, output power of BESS $i$ can be calculated as follow
\begin{equation}
	P_i=\frac{\lambda_i}{2\beta_i}-\frac{\alpha_i}{2\beta_i}.
	\label{7}
\end{equation}
\subsection{Control objectives}
\begin{figure}
	\centering
	\includegraphics[width=7cm]{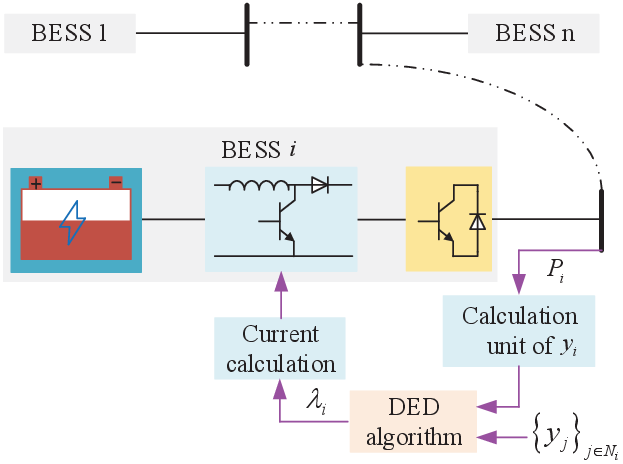} \caption{The control framework with a DED algorithm for the networked BESSs.\label{figei1}}
\end{figure}
\par In this paper, we aim to reduce the management expenses of networked BESSs as shown in Fig. \ref{figei1}, including inner power loss cost and capacity degradation cost. Meanwhile, the balance of supply and demand of power should be maintained. Specifically, they are described as follows:
\begin{itemize}
	\item In the island BESSs, MCs reach consensus so that total active power load in the resistive network is shared among batteries according to MCs. That is, mathematically,
	$$\lim_{t \to +\infty}{\left|\lambda_i-\lambda_j\right|}=0,$$
	where $i,j\in {1,2,\cdots,n}$.
	\item Furthermore, a balance needs to be maintained between the sum of the local power determined by MCs and the total active power load., i.e.,
	$$\sum_{i=1}^{n}P_i=\sum_{i=1}^{n}D_i.$$
\end{itemize}
\begin{remark}
	The inner power loss ratio $\gamma_i$ only conforms to (\ref{2}) on a small time scale. The ED algorithm based on MCs consensus has promising applications in the secondary control of BESSs. Combining MC with droop control, it is easy to design a secondary controller with MC consensus. In this way, BESS can simultaneously achieve optimal ED, frequency (voltage) restoration, and supply-demand balance in AC (DC) MG on a small time scale.
\end{remark}
\section{A DED scheme with reset mechanism for BESS}
\par The DED scheme has been designed in previous studies based on MC consensus both in the grid-connected and isolated mode. Unlike BESSs in the grid-connected mode, the total load in the network needs to be shared internally by all isolated BESSs. Unfortunately, currently designed distributed schemes are based on PCs or finite/fixed time controllers, both of which can accomplish tasks with significant disadvantages in terms of dynamic performance. For example, PCs cause the system to suffer from low control accuracy, while the introduction of fractional terms in fixed/finite time protocols brings about a negative effect of chattering for MC and power mismatches. In view of this, a PI controller with a reset mechanism based DED scheme is proposed such that MCs reach consensus faster and smoother \textcolor{blue}{than traditional schemes} while the balance between supply and demand is maintained in real time.
\subsection{Graph theory}
\par Generally, Laplacian matrix is indispensable in the controller based on multi-agent system, which often corresponds to a communication topology. For $n$ BESSs governed by \textcolor{blue}{an} MAS, agents communicate with each other over a communication graph $\cal G(\cal V,\cal E)$ with a set of vertices ${\cal V}=\{v_1,v_2,\cdots,v_n\}$ and a set of edges ${\cal E}=\{(v_i,v_j)|v_i,v_j\in{\cal V}\}$, where $v_i$ and $v_j$ are viewed as \textcolor{blue}{agents} $i$ and $j$ respectively. Denote $A=[a_{ij}]$ and $D=diag\{d_1,d_2,\cdots,d_n\}$ as the adjacency matrix and degree matrix corresponding to the graph $\cal G$ respectively. $(v_i,v_j)$ implies that agent $i$ can access to agent $j$, and $v_j$ is viewed as a neighbor of agent $i$ and $a_{ij}=1$. If $v_j$ is not a neighbor of agent $i$, $a_{ij}=0$. $d_i$ is the number of neighbors of agent $i$. Thus, the Laplacian matrix $L$ can be calculated by $L=D-A$.   
\begin{assumption}
	The communication graph of the MAS is undirected and connected.\label{assum 1}
\end{assumption} 
\begin{lemma}
	\label{lemma1}
	\cite{linNecessarySufficientGraphical2005}: For an undirected connected graph, $L$ is a positive semi-definite matrix, whose eigenvalues are non-negative real and ordered as $0=\eta_1\le \eta_2\le \cdots\le \eta_n$.
\end{lemma}
\subsection{A DED algorithm with a zero crossing based reset mechanism for isolated BESSs} 
\par Compared with a P controller based DED scheme, MCs and power mismatches under P controllers \textcolor{blue}{exhibit} good response speed, but low control accuracy. By introducing integral control to form a PI controller, the system behaves good control accuracy, but overshoot appears. To overcome the above shortcomings, a reset mechanism, as shown in Fig. \ref{pir}, is introduced here to modify the standard PI controller so that MCs can reach consensus quickly without overshoot. At the same time, power demand and supply balance can be maintained. 
\par Thus, we first recall a P control scheme in \cite{yuFrequencySynchronizationPower2021b} as follow
\begin{subequations}
	\begin{equation}
		\label{8a}
		\dot y_i=u_i,
	\end{equation}
    \begin{equation}
	    \label{8b}
	    u_i=-A_iy_i-B_i\sum_{j=1}^{n}a_{ij}(y_i-y_j),
    \end{equation}
\end{subequations}
where $y_i=[\lambda_i\quad\rho_i]^T$, $\rho_i$ is namely the mismatch power of BESS $i$ with $\rho_i(0)=D_i-P_i(0)$, $A_i=\left[\begin{matrix}
	0 & -1\\
	0& \frac{1}{2\beta_i}
\end{matrix}\right]$ and $B_i=\left[\begin{matrix}
1 & 0\\
-\frac{1}{2\beta_i}& 1
\end{matrix}\right]$. Thus, a compact form can be obtained that
\begin{subequations}
	\begin{equation}
		\label{9a}
		\dot y=u,
	\end{equation}
	\begin{equation}
		\label{9b}
		u=-Hy,
	\end{equation}
\end{subequations}
where $H=A+B(L\otimes I_2)$, $A=diag(A_i)$, $B=diag(B_i)$, $y=[y_1^T,y_2^T,\cdots,y_n^T]^T$, $u=[u_1^T,u_2^T,\cdots,u_n^T]^T$. Define a variable as
\begin{equation}
	\label{10}
	\xi=Hy.
\end{equation}
The relevant properties of the matrix $H$ are given in Lemma \ref{lemma2}.
\begin{lemma}
	\label{lemma2}
	Under Lemma \ref{lemma1}, the matrix $H$ defined in (\ref{10}) possesses a zero and $2n-1$ positive eigenvalues and is diagonalizable.
\end{lemma}
\begin{proof}
	According to \cite{yuFrequencySynchronizationPower2021b}, there exists an orthogonal matrix $W$ such that 
	$$W^{-1}LW=diag(\eta_1,\eta_2,\cdots,\eta_n).$$ Further, 
	$$\begin{array}{l}
		{\rm{    }}({W^{ - 1}} \otimes {I_2})H(W \otimes {I_2})\\
		= diag({A_1} + {\eta _1}B_1,{A_2} + {\eta _2}B_2, \cdots ,{A_n} + {\eta _n}B_n),
	\end{array}$$
	in which the determinant and the characteristic polynomial of block $i$ are as follow 
	$$det(A_i+\eta_iB_i)=\eta_i^2\ge 0,$$
	$$\zeta_i^2-(2\eta_i+\frac{1}{2\beta_i})\zeta_i+\eta_i^2$$
	such that $(A_i+r_1\eta_1B_i)$ has two different nonnegative eigenvalues $\zeta_{i1}$ and $\zeta_{i2}$, i.e., $\zeta_{ij}=\frac{(2\eta_i+\frac{1}{2\beta_i})\pm \sqrt{(2\eta_i+\frac{1}{2\beta_i})^2-4\eta_i^2}}{2}$ with $j\in\{1,2\}$.
	Given that $\eta_1=0$, if and only if $i=1$, $det(A_1+\eta_1B_1)=0$ and $rank(A_1+\eta_1 B_1)=1$. Thus, $rank(H)=2n-1$ and 0 is a simple eigenvalue of the matrix $H$ with algebraic multiplicity 1.
Thus, there must exist an invertible matrix $P_i$ such that 
	$$P_i^{-1}(A_i+r_1\eta_1I_2)P_i=diag(\zeta_{i1}, \zeta_{i2}).$$ 
	Denote $P=diag(P_1,P_2,\cdots,P_n)$. From this, it can be concluded that there must be an invertible matrix $(W\otimes I_2)P$ such that $((W\otimes I_2)P)^{-1}H(W\otimes I_2)P$ is a diagonal matrix. Then, Lemma \ref{lemma2} is proved.$\hfill\blacksquare$
\end{proof}
\par Thus, we design a DED scheme with a reset mechanism as follow
\begin{subequations}
	\begin{equation}
		\label{11a}
		u=-r_1\xi-r_2\nu,
	\end{equation}
	\begin{equation}
		\label{11b}
		\left\{\begin{array}{l}
			\dot \nu_i =\xi_i, \,\,\quad if\quad (\xi_i,t)\in \bigcap\limits_{j = 1}^2 {\cal F}_{ij}\\
			\nu _i ({t^ + }) = \Gamma_i\nu_i, if\quad (\xi_i,t)\in \bigcup\limits_{j = 1}^2 {\cal J}_{ij},
		\end{array} \right.
	\end{equation}
\end{subequations}
where $r_1,r_2>0$ are control gains, $\nu=[\nu_1,\nu_2,\cdots,\nu_n]$, ${\cal F}_{ij}=\{(\xi_{ij},t)|\xi_{ij}\neq 0\}$ and ${\cal J}_{ij}=\{(\xi_{ij},t)|\xi_{ij}=0\}$ are the flow set and jump set, $\Gamma_i=[\gamma_{ij}]$ and $\gamma_{ij}=\left\{\begin{array}{l}
	1,\,\,if\quad \xi_{ij}\in {\cal F}_{ij}\\
	0,\,\,if\quad \xi_{ij}\in {\cal J}_{ij}
\end{array}\right.$, (\ref{11b}) is so-called a reset mechanism based on the Clegg integrator \cite{cleggNonlinearIntegratorServomechanisms1958}. Denote $Y=[y\quad\nu]^T$. Thus, a compact form of the base system is
\begin{equation}
	\label{12}
	\dot Y=-\Phi Y,
\end{equation}
where $\Phi=\left[\begin{matrix}
	r_1H & r_2I \\
	-H & 0
\end{matrix}\right]$. The designed DED scheme with a PI+R controller for isolated BESSs is shown in Fig. \ref{pir}.
\begin{remark}
		\label{rem2}
		The essence of controller design is that, as described in Fig. \ref{pir}, when the proportional term undergoes a zero crossing, the integral term is reset to 0. In this way, the integral term always maintains the same sign as the proportional term. On the one hand, the introduction of integral terms improves control accuracy. On the other hand, the reset mechanism ensures that $\Vert u_i\Vert\ge\Vert r_1\xi_i\Vert$, thereby accelerating \textcolor{blue}{the system convergence}.
\end{remark}
\begin{figure}
	\centering
	\includegraphics[width=8cm]{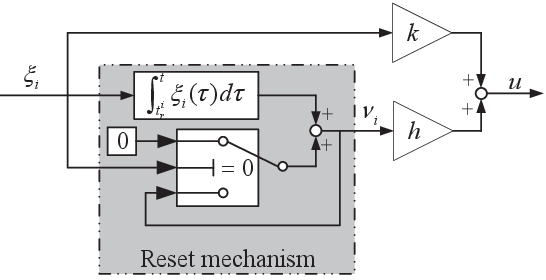} \caption{The designed DED algorithm with reset mechanism\label{pir}}
\end{figure}
\par For the convenience of subsequent analysis, \textcolor{blue}{we} first provide a definition and properties of reset control systems.
\begin{definition}
	\cite{banosResetControlSystems2011} Consider a reset control system 
	$$\left\{\begin{array}{l}
		\dot x=Fx, \qquad x\in \cal M\\
		x(t^+)=F_rx, x\notin \cal M,
	\end{array}\right.$$ 
	where $F$ is a time invariant matrix, $\cal M$ is a reset surface. The first equation is called the base system. If the base system is stable, or unstable and its solution is zero at some instant, the reset control system is recognized to be regular. In particular, if the base system possesses complex dominant modes, the zero crossing based reset control system behave regular with at least one reset instant. 
	\label{Definition 1}
\end{definition}
\par From Definition \ref{Definition 1}, complex dominant modes, i.e., complex dominant eigenvalues, promote the reset mechanism to function.
\par Derivation of (\ref{10}) yields that
\begin{equation}
	\label{13}
	\dot \xi=H\dot y=-r_1H\xi-r_2H\nu
\end{equation}
Next, a compact form is given as follow
\begin{equation}
	\label{14}
	\dot z=-\Phi_0 z,
\end{equation}
where $z=[\xi\quad \nu]^T$, $\Phi_0=\left[\begin{matrix}
	r_1H & r_2H \\
	-I & 0
\end{matrix}\right]$. 
\par Next, the stability of the base system (\ref{12}) is given in Lemma \ref{lemma3} and analyzed in its proof.
\begin{lemma}
	\label{lemma3}
	Under \textcolor{blue}{Lemmas} \ref{lemma1} and \ref{lemma2}, the base system (\ref{12}) is globally asymptotic stable. That is, the system (\ref{11a}) with (\ref{11b}) is regular. And the balance between supply and demand can be well maintained.
\end{lemma}
\begin{proof}
	According to Lemma \ref{lemma2},  the characteristic polynomial of $\Phi$ is
	$$\begin{array}{l}
		\vert \varsigma I-\Phi\vert=\vert \varsigma^2I-r_1H\varsigma+r_2H\vert\\
		\qquad \quad \,\,\,\;=\prod_{i=1}^{n}\prod_{j=1}^{2}(\varsigma^2-r_1\zeta_{ij}\varsigma+r_2\zeta_{ij}),
	\end{array}$$
	which gives two zero eigenvalues and $(4n-2)$ nonzero eigenvalues with positive real parts.
Based on Lemma \ref{lemma1}, it is easy to know that $rank(\Phi)=4n-3$, which means that zero is eigenvalues of $\Phi$ with geometric multiplicity 1 and algebraic multiplicity 2. Thus, there must exist an invertible matrix $S$ such that a Jordan matrix can be obtained as follow
	$$J=S^{-1}\Phi S$$
	\par The solution of (\ref{12}) is $Y=Sexp(-Jt)S^{-1}Y(0)$. Assume that $v_1$ and $v_2$, respectively, are the first and second columns of $S$, and $v_1$ is the vector associated with zero eigenvalue. And it is easily obtained that
	$$v_1=span\{[1,0]^T\otimes 1_n^T\otimes [1,0]^T\}.$$
	Thus, $\lim\limits_{t\to+\infty}Y=[v_1,v_1,0_{4n\times (4n-2)}]S^{-1}Y(0)$. According to $\nu_i(0)=0_2$, one can get $\lim\limits_{t\to+\infty}Y\in span\{v_1\}$. This indicates that the base system (\ref{12}) is globally asymptotic stable. And $\lim\limits_{t\to+\infty} \vert \lambda_i-\lambda_j\vert=0$, $\lim\limits_{t\to+\infty} \nu_i=0_2$ and $\lim\limits_{t\to+\infty} \rho_i=0$ with $i,j\in\{1,2,\cdots,n\}$. According \textcolor{blue}{to} Definition \ref{Definition 1}, the system (\ref{11a}) with (\ref{11b}) is regular.
	\par Noting that $diag([\frac{1}{2\beta_i},1])H=L\otimes[0\;1]$, for any $i\in \{1,2,\cdots,n\}$, we have
	$$\begin{array}{l}
		1_n^T(\dot P+\dot\rho)=1_n^Tdiag([\frac{1}{2\beta_i},1])\dot y\\
		\qquad \qquad\;\;\,=-1_n^Tdiag([\frac{1}{2\beta_i},1])(r_1\xi+r_2\nu)\\
		\qquad \qquad\;\;\,=-1_n^Tr_2diag([\frac{1}{2\beta_i},1])\nu
	\end{array}.$$ 
	Besides, there must be
	$$\begin{array}{l}
		r_21_n^Tdiag([\frac{1}{2\beta_i},1])\dot\nu=r_21_n^Tdiag([\frac{1}{2\beta_i},1])\xi\\
		\qquad\qquad\qquad\qquad\;\;=r_21_n^T(L\otimes[0\;1])y=0
	\end{array},$$
	which implies 
	$$r_21_n^Tdiag([\frac{1}{2\beta_i},1])\nu=r_21_n^Tdiag([\frac{1}{2\beta_i},1])\nu(0)=0.$$
	Therefore, there must be
	$$1_n^T(\dot P+\dot \rho)=0,$$
	which implies $1_n^T(P+\rho)=\sum_{i=1}^{n}D_i$ and $\lim\limits_{t\to+\infty} \sum_{i=1}^{n}P_i=\sum_{i=1}^{n}D_i$. The balance between supply and demand can be well maintained.
	So far, Lemma \ref{lemma3} is proved. $\hfill\blacksquare$
\end{proof}
\par Thus, the parameter conditions for a enabled reset mechanism are given in Theorem \ref{TH1} and analyzed in its proof.
\begin{theorem}
	\label{TH1}
	Under \textcolor{blue}{Lemmas} \ref{lemma1}, \ref{lemma2} and \ref{lemma3}, if $k$, $r_1$ and $r_2$ satisfy (\ref{15}), all nonzero eigenvalues of $\Phi_0$ are complex with positive real parts. At this point, the zero crossing based reset mechanism acts.
\end{theorem}
\begin{proof}
	Looking back at the matrix $\Phi_0$, one can get that
	$$\begin{array}{l}
		\vert\varsigma I-\Phi_0\vert=\vert \varsigma^2I-r_1H\varsigma+r_2H\vert \\
		\qquad \quad \,=\prod_{i=1}^{n}\prod_{j=1}^{2}(\varsigma^2-r_1\zeta_{ij}\varsigma+r_2\zeta_{ij})
	\end{array},$$
	which gives one zero eigenvalue with geometric multiplicity 2 and algebraic multiplicity 1 and $4n-2$ nonzero eigenvalues with negative real parts, i.e., $\frac{r_1\zeta_{ij}\pm \sqrt{r_1^2\zeta_{ij}^2-4r_2\zeta_{ij}}}{2}$. To ensure that $\xi$ based reset mechanism acts, a sufficient condition is given here
	$${r_1^2\zeta_{ij}^2-4r_2\zeta_{ij}}<0$$
That is, 
	\begin{equation}
		\label{15}
		\begin{array}{l}
			\frac{4r_2}{r_1^2}>max\{\zeta_{ij}\}=\frac{(2\eta_i+\frac{1}{2\beta_i})+ \sqrt{(2\eta_i+\frac{1}{2\beta_i})^2-8\eta_i^2}}{2}
		\end{array}.
	\end{equation}
	Accordingly, MCs can reach consensus with at least one reset instant under (\ref{15}). That is, the system is regular and the reset mechanism acts.$\hfill\blacksquare$
\end{proof}
\par For the system (\ref{11a}) and (\ref{11b}) with an enabled zero crossing based reset mechanism, the stability is reanalyzed in Theorem \ref{TH2} and its proof.
\begin{theorem}
	\label{TH2}
	Under \textcolor{blue}{Lemmas} \ref{lemma1}, \ref{lemma2} and \ref{lemma3}, MCs and locally power mismatch, respectively, can reach consensus. 
\end{theorem}
\begin{proof}
	Define a candidate Lyapunov function as follow
	$$V=\mathop{max}\limits_{1\le i\le n}(\lambda_i )-\mathop{min}\limits_{1\le i\le n}(\lambda_i )+\mathop{max}\limits_{1\le i\le n}(\rho_i)-\mathop{min}\limits_{1\le i\le n}(\rho_i).$$
	Naturally, there must be $V\ge 0$ with equality if \textcolor{blue}{and} only if $\lambda_1=\cdots=\lambda_n$ and $\rho_1=\cdots=\rho_n$.
	\par The set-value lie derivative is defined as follow
	$${\cal{L}}_y V=\partial V^T u,$$
	where $\partial V$ is called as the generalized of $V$. Thus, 
	$$\partial V=M-m,$$
	where
	$$M=[M_{1 1},M_{1 2},M_{2 1},M_{2 2},\cdots,M_{n 1},M_{n 2}]^T,$$
	$$m=[m_{1 1},m_{1 2},m_{2 1},m_{2 2},\cdots,m_{n 1},m_{n 2}]^T,$$
	\textcolor{blue}{$M_{i1}=1$, $M_{i2}=1$, $m_{i1}=1$, and $m_{i2}=1$ imply that $\xi_{i1}=max\{\xi_{11},\cdots,\xi_{n1}\}$, $\xi_{i2}=max\{\xi_{12},\cdots,\xi_{n2}\}$, $\xi_{i1}=min\{\xi_{11},\cdots,\xi_{n1}\}$, $\xi_{i2}=min\{\xi_{12},\cdots,\xi_{n2}\}$, respectively.} Given that $\nu_i$ is aligned with $\xi_i$, we have
	$${\cal{L}}_y V=\sum_{i=1}^{n}[m_{i1}-M_{i1}\ m_{i2}-M_{i2}] (r_1 \xi_i+r_2 \nu_i)\le 0.$$
	\par According to \textcolor{blue}{the} LaSalle's invariance principle, $y$ converges asymptotically to a strongly invariant set
	$$\{y_i|y_1=y_2=\cdots=y_n\}.$$
	At this point, $\xi_i$ and $\nu_i$ converge asymptotically to 0 for $i\in\{1,2,\cdots,n\}$. Besides, according to (\ref{lemma3}), the solution of (\ref{12}) is $\lim\limits_{t\to+\infty}Y\in span\{[1,0]^T\otimes 1_n^T\otimes [1,0]^T\}$. Thus, MCs and power mismatches, respectively, can reach consensus. Further, power mismatches asymptotically converge to 0. 
	So far, Theorem \ref{TH2} is proved. $\hfill\blacksquare$
\end{proof}
\section{Further disscuss}
\par The above only discusses the feasibility of the distributed PI+R solution to solve the convergence of MC and power mismatch. Next, we will discuss some improvement measures to address power supply and demand balance, reset Zeno behavior and time delay, and capacity limitations.
\subsection{Discussion on supply-demand balance}
\par According to Theorem \ref{TH2}, MCs and power mismatches of BESSs can reach consensus respectively. The supply-demand balance has not been accurately analyzed. Also, it can not be guaranteed by the above design that the power mismatch converges to 0. Next, the balance between supply and demand in the flow and the jump sets is illustrated respectively. Firstly, the following equation holds at any flow instant $\{t_r<t\le t_{r+1}\}$
$$\begin{array}{l}
	1_n^T(\dot P+\dot \rho)=-1_n^Tdiag([\frac{1}{2\beta_i},1])(r_1\xi+r_2\nu)\\
	\qquad\qquad=-r_21_n^Tdiag([\frac{1}{2\beta_i},1])\int_{t_r}^{t}{\xi d\tau}\\
	\qquad\qquad=-r_21_n^Tdiag([\frac{1}{2\beta_i},1])\nu(t_r^+)
\end{array}.$$
Secondly, at the first reset instant $t_1$, one can get
$$\begin{array}{l}
	1_n^T(\dot P+\dot \rho)=-r_21_n^Tdiag([\frac{1}{2\beta_i},1])\nu(t_1)=0
\end{array},$$
and
\begin{equation}
	\label{16}
	\lim\limits_{t\to +\infty}1_n^T(P+\rho)=D-r_21_n^Tdiag([\frac{1}{2\beta_i},1])\sum_{r=1}^{+\infty}\nu(t_r) (t_{r+1}-t_r).
\end{equation}
Under the reset mechanism, one can get
$$r_21_n^Tdiag([\frac{1}{2\beta_i},1])\nu(t_r^+)\neq 0,$$
which is because at the reset instant, one or some individuals are reset, resulting in inaccurate power mismatch estimation at this instant. At any subsequent moment, power mismatches also cannot be guaranteed accurate estimation. Thereby, the supply and demand balance is unknown.
\begin{remark}
	\label{remark2}
	Due to the reset mechanism, it is difficult to mathematically describe the solution of (\ref{12}) in detail. Undoubtedly, it can be seen that the reset mechanism damages the balance of supply and demand. Thus, the ED problem of isolated BESSs cannot be effectively solved under this scheme. As described in the proof of Theorem \ref{TH2}, the final state of a cluster system cannot be accurately described when only one or some individuals are reset at certain instants. Similar issue is also described in \cite{mengResetControlSynchronization2019}. In view of this, an improved reset mechanism is needed. The simplest and most direct method is centralized, as shown below,
	\begin{equation}
		\label{17}
		\left\{\begin{array}{l}
			\dot \nu_i =\xi_i, \,\,\quad if\quad (\xi_i,t)\in \bigcap\limits_{j = 1}^2 {\cal F}_{ij}\\
			\nu ({t^ + }) = 0_{2n}, if\quad (\xi_i,t)\in \bigcup\limits_{j = 1}^2 {\cal J}_{ij},
		\end{array} \right..
	\end{equation}
	Hence, in each flow set, $\rho$ gradually approaches $0_n$. That is, $\rho$ asymptotically \textcolor{blue}{converges} to $0_n$. Combined with (\ref{16}), the balance between power supply and demand can be well maintained.
\end{remark}
\subsection{Discussion on the Zeno behavior and time delay}
\par To avoid the Zeno behavior, a dwell time is given here, i.e.,
\begin{subequations}
	\begin{equation}
		\label{18a}
		\dot u=-r_1\xi-r_2\nu,
	\end{equation}
	\begin{equation}
		\label{18b}
		\left\{\begin{array}{l}
			\dot \nu_i =\xi_i, \,\,\qquad\,\,\, if\quad (\xi_i,t,T_i)\in \bigcap\limits_{j = 1}^2 {\cal F}_{ij}\\
			\nu ({t^ + }) = 0_{2n}, if\quad (\xi_i,t,T_i)\in \bigcup\limits_{j = 1}^2 {\cal J}_{ij},
		\end{array} \right.
	\end{equation}
	\begin{equation}
		\label{18c}
		\left\{\begin{array}{l}
			\dot T_i=1, \,\,\qquad\,\,\,\,\, if\quad (\xi_i,t,T_i)\in \bigcap\limits_{j = 1}^2 {\cal F}_{ij}\\
			T_i(t^+)=\Gamma_iT_i, if\quad (\xi_i,t,T_i)\in \bigcup\limits_{j = 1}^2 {\cal J}_{ij},
	    \end{array} \right.
	\end{equation}
\end{subequations}
where ${\cal F}_{ij}=\{(\xi_{ij},t)|\xi_{ij}\neq 0\vee T_{ij}<\rho_{ij}\}$, ${\cal J}_{ij}=\{(\xi_{ij},t)|\xi_{ij}=0\land T_{ij}\ge\rho_{ij}\}$, $\Gamma_i=[\gamma_{ij}]$, $\gamma_{ij}=\left\{\begin{array}{l}
	1,\,\,if\quad \xi_{ij}\in {\cal F}_{ij}\\
	0,\,\,if\quad \xi_{ij}\in {\cal J}_{ij}
\end{array}\right.$.
\par Avoiding Zeno behavior not only improves the practicality of the algorithm, but also alleviates the stability problems caused by input delay.
\begin{remark}
	\label{remark3}
	Consider an uniform delay $\tau$. Thus, applying the Artstein’s transformation \cite{zhangDistributedFiniteTimeMultiagent2019, artsteinLinearSystemsDelayed1982}, we have
	\begin{equation}
		\dot y=u(t-\tau)
	\end{equation}
	Applying the Artsteins transformation, we have
	\begin{equation}
		\dot y^\tau=u,
	\end{equation}
	where $y^\tau=y+\int_{t-\tau}^{t}u^\tau(s)ds$. There are two conditions that need to be met: (1) the initial state $y_i(0)$ for $\forall i \in \{1,2,\cdots,n\}$ is bounded, and (2) $u_i$ for $\forall i \in \{1,2,\cdots,n\}$ is integrable. Due to the stable output power of a normally operating MG, the initial state is naturally guaranteed to be bounded. Due to the reset mechanism with a dwell time excluding Zeno behavior, \textcolor{blue}{the} Riemannian integrability is also guaranteed.
\end{remark}
\subsection{Discussion on the capacity limitation}
\par Consider that each BESS has the capacity limitation $P_{m,i}\le P_i\le P_{M,i}$ for $i\in\{1,2,\cdots,n\}$, where $P_{m,i}$ and $P_{M,i}$ are the upper and lower power limits of BESS $i$, respectively. Thus, considering \eqref{7} for BESS $i$ at some instant, we have
\begin{equation}
	\label{pl}
	P_i=\left\{\begin{array}{l}
		P_{m,i},\;if\;\frac{\lambda_i}{2\beta_i}-\frac{\alpha_i}{2\beta_i}<	P_{m,i},\\
		\frac{\lambda_i}{2\beta_i}-\frac{\alpha_i}{2\beta_i},\;if\;P_{m,i}\le \frac{\lambda_i}{2\beta_i}-\frac{\alpha_i}{2\beta_i}\le P_{M,i},\\
		P_{M,i},\;if\;\frac{\lambda_i}{2\beta_i}-\frac{\alpha_i}{2\beta_i}>P_{M,i}.
	\end{array}\right.
\end{equation}
Define a set that ${\cal M}=\{i|P_{m,i}\le \frac{\lambda_i}{2\beta_i}-\frac{\alpha_i}{2\beta_i}\le P_{M,i},\;i\in\{1,2,\cdots,n\}\}$. During the discharge process of BESS, $P_{m,i}\ge 0$ and $P_{M,i}\ge 0$. $P_{m,i}=P_{M,i}=0$ implies that bus $i$ is a pure load node in the network.
\par To ensure power supply and demand balance under capacity constraints, we propose a revised MC scheme by time trigger as follow.
	\par Firstly, we define the following variables
\begin{subequations}
	\label{22}
	\begin{equation}
		\Delta P_i=\left\{\begin{array}{l}
			\frac{\lambda_i-\beta_i}{2\alpha_i}-{\tilde P}_i,\;if\;i\notin{\cal M},\\
			0,\;\mathrm{otherwise},
		\end{array}\right.
	\end{equation}
	\begin{equation}
		\Delta \alpha_i=\left\{\begin{array}{l}
			0,\;if\;i\notin{\cal M},\\
			\frac{1}{2\alpha_i},\;\mathrm{otherwise}.
		\end{array}\right.
	\end{equation}
\end{subequations}
In a compact form, $\Delta P_i$ and $\Delta \alpha_i$ for $i\in\{1,2,\cdots,n\}$ \textcolor{blue}{reach} consensus, respectively, by
\begin{equation}
	\label{24}
	\left\{\begin{array}{l}
		\Delta \dot P=-h_1\xi_P-h_2\nu_P\\
		\Delta \dot\alpha=-h_1\xi_\alpha-h_2\nu_\alpha
	\end{array}\right.
\end{equation}
where $\xi^P=L\Delta P$, $\xi^\alpha=L\Delta \alpha$,
\begin{subequations}
	\begin{equation}
		\left\{\begin{array}{l}
			\nu_P^{+}=0_n,\;if\;\xi_{P,i}\nu_{P,i}<0\;for\;\exists i,\\
			\dot \nu_P=\xi_P,\;otherwise,
		\end{array}\right.
	\end{equation}
	\begin{equation}
		\left\{\begin{array}{l}
			\nu_\alpha^{+}=0_n,\;if\;\xi_{\alpha,i}\nu_{\alpha,i}<0\;for\;\exists i,\\
			\dot \nu_\alpha=\xi_\alpha,\;otherwise.
		\end{array}\right.
	\end{equation}
\end{subequations}
In this way, MC is modified as
\begin{equation}
	\label{27}
	\bar\lambda_i=\lambda_i+\frac{\Delta P_i}{\Delta \alpha_i},\;i\in\{1,2,\cdots,n\}.
\end{equation}
Accordingly, the output power of BESS can be calculated by
\begin{equation}
	\label{28}
	\left\{\begin{array}{l}
		\tilde P_i=\frac{\bar\lambda_i-\beta_i}{2\alpha_i},\;if\;i\notin{\cal M},\\
		\tilde P_i=P_{m,i}\;\mathrm{or}\;P_{M,i},\;\mathrm{otherwise}.
	\end{array}\right.
\end{equation}
\begin{theorem}
	Under Assumption \ref{assum 1}, if $\frac{4h_2}{h_1^2}>\eta_n$, $\Delta P_i$ and $\Delta\alpha_i$ for $i\in\{1,2,\cdots,n\}$ can reach average consensus with at least one reset instant by \eqref{24}, respectively. Under the above, \eqref{27}, and \eqref{28}, the output power of BESS will not exceed the limit, and the balance between supply and demand can be maintained.
\end{theorem}
\begin{proof}
	According to the previous analysis method, the average consensus of $\Delta P_i$ for $i\in\{1,2,\cdots,n\}$ generated from the last monitoring can be guaranteed by \eqref{24}, \textcolor{blue}{and the same applies to $\Delta\alpha_i$}. According to the result of \cite{Chen2018} and the above, the output power of each BESS does not exceed the limit, the balance between supply and demand can gradually be ensured by \eqref{27} and \eqref{28}.$\hfill\blacksquare$
\end{proof}
\begin{remark}	
	Looking back at the definition of MC, i.e., $\lambda_i=2\beta_iP_i+\alpha_i$ with $\beta_i>0$ and $\alpha_i>0$ for $i\in\{1,2,\cdots,n\}$, we can find that MC is positively linearly correlated with power. That is to say, as the power increases/decreases, the MC also increases/decreases accordingly. This may lead to inconsistency in the actual MC. Equivalently, an MC consensus scheme can be designed, and the power calculated using \eqref{pl} needs to meet the supply-demand balance. It is worth noting that the consensus value of MCs in this sense is different from the scheme without capacity constraints. How to ensure MC consensus in the latter solution is a question worth studying.
\end{remark}
\section{Case study}
Consider \textcolor{blue}{an} MG containing $4$ BESSs governed by \textcolor{blue}{an} MAS. \textcolor{blue}{The single line diagram of this BESS network and their equipped communication topology are shown in Fig. \ref{figc}}. Accordingly, $\eta_n(L)=4$. Four cases are designed here to test the fully distributed scheme, the distributed scheme with a centralized reset mechanism. Further, the designed Zeno-free controllers with and without input delays are \textcolor{blue}{simulated, respectively}. The details are as follows.
\begin{figure}
	\centering
	\includegraphics[width=8cm]{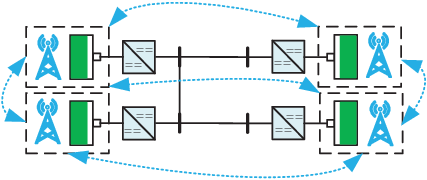} \caption{Communication graph topology.\label{figc}}
\end{figure}
\par In Case 1, the effectiveness of the proposed fully DED controller (\ref{11a}) with a reset mechanism (\ref{11b}) is verified under $r_1=1$ and $r_2=1.5$. At this point, the inequality (\ref{15}) in Theorem \ref{TH1} holds.
 \par In Case 2, the effectiveness of the proposed DED controller (\ref{11a}) with a centralized reset mechanism (\ref{17}) designed in Remark \ref{remark2} \textcolor{blue}{is} verified under $r_1=1$, $r_2=1.5$ and $T_{ij}=0$.
\par In Case 3, the effectiveness of the proposed Zeno-free DED controller (\ref{18a}) with a centralized reset mechanism (\ref{18b}) and (\ref{18c}) \textcolor{blue}{is} verified under $r_1=1$, $r_2=1.5$ and $T_{ij}=0.5$.
\par In Case 4, based on the simulation in Case 3, a delay vector $\tau=[0.3,0.4,0.5,0.45]^T$ is introduced under $r_1=1$, $r_2=1.5$ and $T_{ij}=0.5$ to verify the designed reduced method in Remark \ref{remark3}.
\par In Case 5, considering capacity limits $P_M=[6\;2.6\;4\;7]^T$ and $P_m=[0\;0\;0\;0]^T$, the modified MC scheme is tested to meet supply and demand balance. 
\par In Case 6, based on Case 5, the plug and play function is tested, which indirectly verifies that the designed \textcolor{blue}{scheme} can solve the ED problem of networks with pure load nodes.
\subsection{Case 1: Validation of the distributed PI+R scheme based on zero crossing}
\par This case is arranged to validate the fully distributed PI+R scheme, where the zero crossing based distributed reset mechanism is adopted. The simulation results are shown in Fig. \ref{fig2} and \textcolor{blue}{Fig.} \ref{fig3}.
\begin{figure}
	\centering
	\subfigure[Evolution of MC and power mismatch of each BESS\label{fig1a}]{\includegraphics[width=8cm]{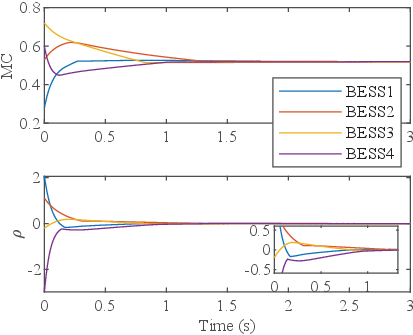}}
	\subfigure[Evolution of integral term of each BESS and total output power\label{fig1b}]{\includegraphics[width=8cm]{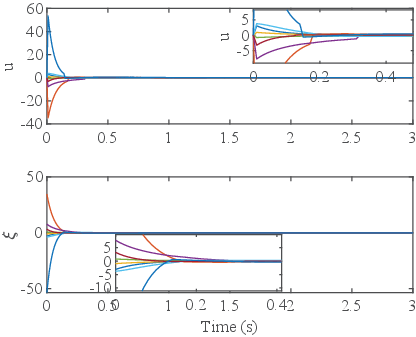}}
	\caption{Evolution of MC, power mismatch integral term of each BESS and total output power in Case 1\label{fig2}}
\end{figure}
\begin{figure}
	\centering
	\subfigure[Evolution of integral term of each BESS\label{fig1c}]{\includegraphics[width=8cm]{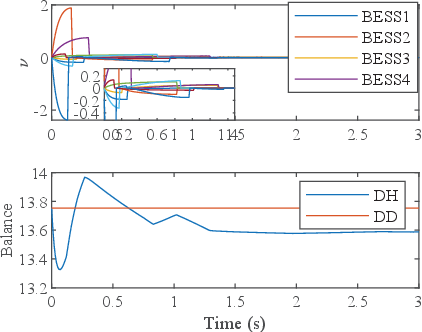}}
	\subfigure[Evolution of output power of each BESS and total output power\label{fig1d}]{\includegraphics[width=8cm]{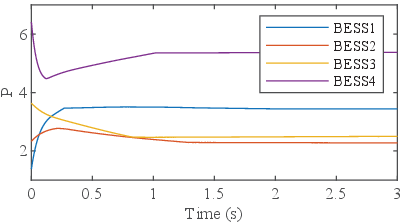}}
	\caption{Evolution of integral term, output power of each BESS in Case 1\label{fig3}}
\end{figure}
\par From Fig. \ref{fig2}, MCs can achieve consensus, while power mismatch, control error and input of each BESS can reach 0. From Fig. \ref{fig1c}, each integral term is reset to 0 several times. Although the output power of each BESS is stable, as shown in Fig. \ref{fig1d}, the imbalance between total output power and demand is obvious. This method is not applicable to isolated BESSs.
\subsection{Case 2: Validation of the PI+R scheme based on a centralized approach}
\par To adapt to isolated BESSs, a centralized method in (\ref{17}) is introduced and simulated in this case. The simulation results are shown in Fig. \ref{fig4} and \ref{fig5}.
\par From Fig. \ref{fig4}, MCs can achieve consensus, while power mismatch, control error and input of each BESS can reach 0. From Fig. \ref{fig2c}, each integral term is reset to 0 several times, and the output power of each is stable as shown in Fig. \ref{fig2d}. \textcolor{blue}{Comparing} Fig. \ref{fig2c} with Fig. \ref{fig1c} in Case 1, the balance between total output power and demand is well maintained under this centralized method. This indicates that this scheme can solve the ED problem for BESSs while ensuring supply and demand balance.
\par To illustrate the progressiveness, the designed P controller in \cite{yuFrequencySynchronizationPower2021b} is used for comparison, and its simulation results are shown in Fig. \ref{fig41}. From the simulation results of the two schemes, it can be seen that the designed PI+R scheme exhibits obvious characteristics of fast convergence speed and high control accuracy. This validates the discussion in Remark \ref{rem2}.
\begin{figure}
	\centering
	\subfigure[Evolution of MC and power mismatch of each BESS\label{fig2a}]{\includegraphics[width=8cm]{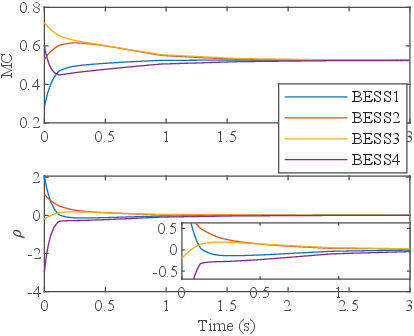}}
	\subfigure[Evolution of integral term of each BESS and total output power\label{fig2b}]{\includegraphics[width=8cm]{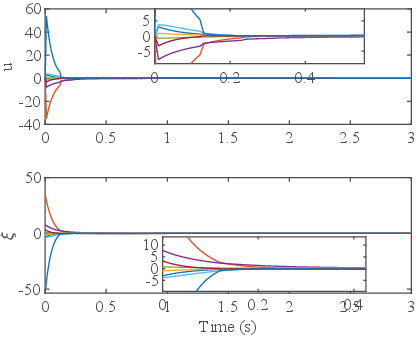}}
	\caption{Evolution of MC, power mismatch, integral term of each BESS and total output power in Case 2\label{fig4}}
\end{figure}
\begin{figure}
	\centering
	\subfigure[Evolution of integral term of each BESS\label{fig2c}]{\includegraphics[width=8cm]{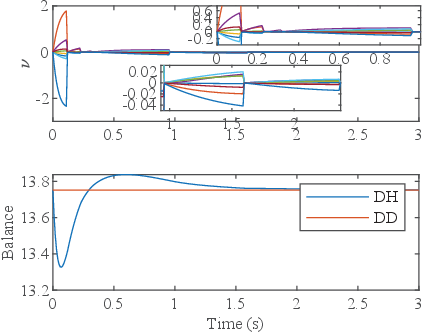}}
	\subfigure[Evolution of output power of each BESS and total output power\label{fig2d}]{\includegraphics[width=8cm]{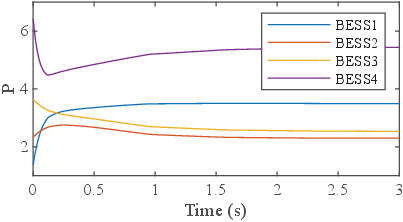}}
	\caption{Evolution of integral term, output power of each BESS in Case 2\label{fig5}}
\end{figure}
\begin{figure}
	\centering
	\includegraphics[width=8cm]{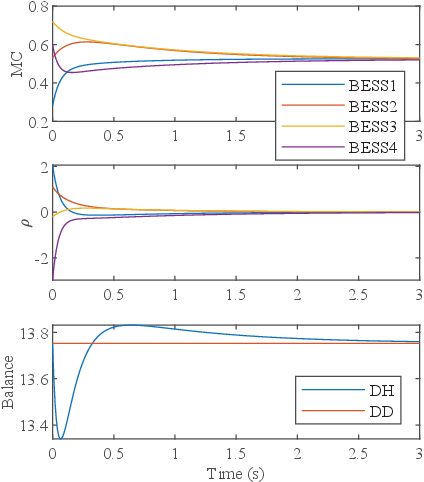}
	\caption{Evolution of MC, power mismatch, and total output power under the control scheme in \cite{yuFrequencySynchronizationPower2021b} in Case 2\label{fig41}}
\end{figure}
\subsection{Case 3: Validation of the designed Zeno-free PI+R scheme based on a centralized approach}
\par To verify the proposed Zeno-free scheme, a minimum dwell time $T_{ij}=0.5$ is introduced. The simulation results are shown in Fig. \ref{fig6} and \textcolor{blue}{Fig.} \ref{fig7}. 
\par It can be observed from Fig. \ref{fig6} that MCs can reach consensus, and power mismatch, control input and error of each BESS converge to $0$. From Fig. \ref{fig7}, integral term of each BESS converges to $0$, while the supply and demand balance achieves, and the output power of each BESS is stable.
\par \textcolor{blue}{Comparing} Fig. \ref{fig7} with Fig. \ref{fig5}, reset instants of each integral term are delayed, and the total output overshoot increases under $T_{ij}=0.5$. To ensure the dynamic performance of the system, the dwell time $T_{ij}$ should not be too large.
\begin{figure}
	\centering
	\subfigure[Evolution of MC and power mismatch of each BESS\label{fig3a}]{\includegraphics[width=8cm]{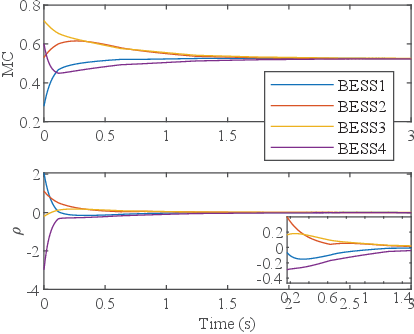}}
	\subfigure[Evolution of integral term of each BESS and total output power\label{fig3b}]{\includegraphics[width=8cm]{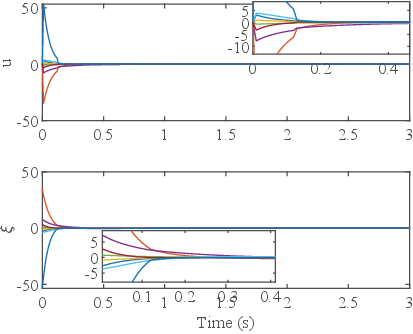}}
	\caption{Evolution of MC, power mismatch integral term of each BESS and total output power in Case 3\label{fig6}}
\end{figure}
\begin{figure}
	\centering
	\subfigure[Evolution of integral term of each BESS\label{fig3c}]{\includegraphics[width=8cm]{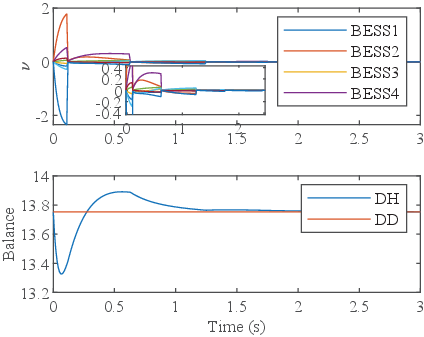}}
	\subfigure[Evolution of output power of each BESS and total output power\label{fig3d}]{\includegraphics[width=8cm]{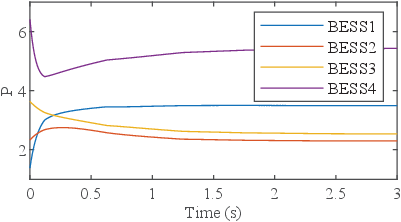}}
	\caption{Evolution of integral term, output power of each BESS in Case 3\label{fig7}}
\end{figure}
\subsection{Case 4: Validation of the PI+R scheme under input delay based on a centralized approach}
\par To test the proposed Zeno-free scheme with time delays, a time delay vector $\tau=[0.3,0.4,0.5,0.45]^T$ is introduced. The simulation results are shown in Fig. \ref{fig8} and \ref{fig9}. 
\par It can be observed from Fig. \ref{fig8} that despite some oscillations, MCs can reach consensus, and power mismatch, control input and error of each BESS converge to $0$. From Fig. \ref{fig9}, integral term of each BESS converges to $0$ with several reset instants, and the supply and demand balance achieves, and the output power of each BESS is stable.
\par \textcolor{blue}{Comparing} Fig. \ref{fig9} with Fig. \ref{fig7}, the oscillation amplitude of each integral term, the output power of each BESS and the total output power increase. It is worth noting that although the total output power can converge to the total load faster, its oscillation amplitude increases and shows negative values. Further improvements may be needed to better address the ED problem of isolated BESSs
\begin{figure}
	\centering
	\subfigure[Evolution of MC and power mismatch of each BESS\label{fig4a}]{\includegraphics[width=8cm]{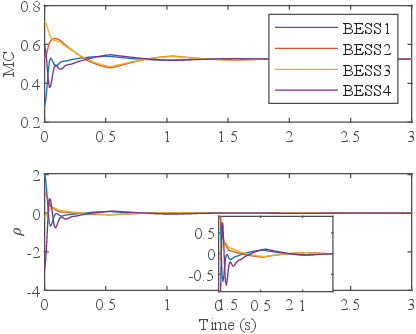}}
	\subfigure[Evolution of integral term of each BESS and total output power\label{fig4b}]{\includegraphics[width=8cm]{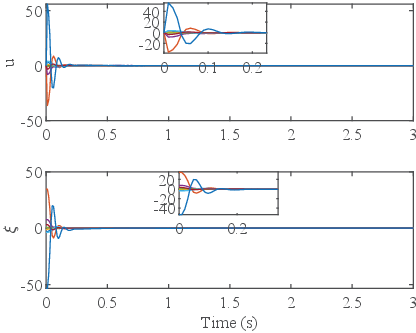}}
	\caption{Evolution of MC, power mismatch integral term of each BESS and total output power in Case 4\label{fig8}}
\end{figure}
\begin{figure}
	\centering
	\subfigure[Evolution of integral term of each BESS\label{fig4c}]{\includegraphics[width=8cm]{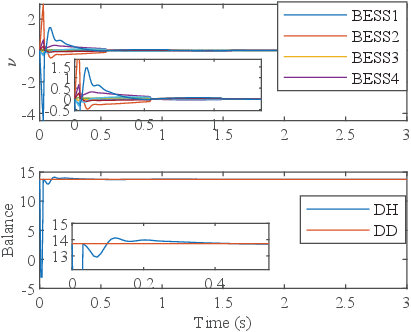}}
	\subfigure[Evolution of output power of each BESS and total output power\label{fig4d}]{\includegraphics[width=8cm]{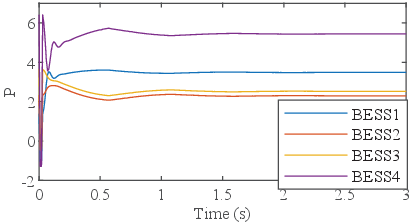}}
	\caption{Evolution of integral term, output power of each BESS in Case 4\label{fig9}}
\end{figure}
\subsection{Case 5: Validation of the PI+R scheme considering capacity limitations based on a centralized approach}
\par Based on the non capacity constrained \textcolor{blue}{version}, we arrange this case to investigate the designed MC correction \textcolor{blue}{solution}, in order to test the power supply and demand balance driven by the corrected MC scheme.
\par From the simulation results in Fig. \ref{fig10}, it can be seen that although the modified scheme caused inconsistency in the MC, the output power of each BESS does not exceed the limit. At the same time, the balance between power supply and demand can also be well maintained. The two variables $\Delta P_i$ and $\Delta\alpha_i$ for $i\in\{1,2,\cdots,n\}$ defined can also remain consensus. It is worth mentioning that the simulation results also indicate that after monitoring and updating the variables at $t=4s$, the supply-demand balance can be asymptotically achieved. In order to illustrate the progressiveness, we compared the schemes in \cite{Chen2018}.
\par From Fig. \ref{fig11}, it can be seen that after monitoring at $t=4s$, the supply-demand balance cannot be achieved. At $t=7s$, another monitoring is conducted, and the supply-demand balance gradually can be maintained thereafter.
\begin{figure}
	\centering
	\subfigure[Evolution of MC, power of each BESS, and the total power\label{fig5a}]{\includegraphics[width=8cm]{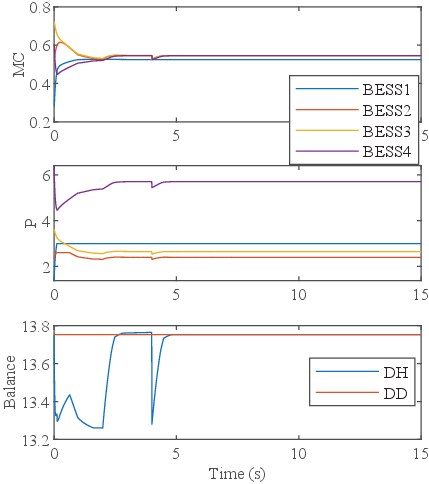}}
	\subfigure[Evolution of $\Delta P_i$ and $\Delta\alpha_i$\label{fig5b}]{\includegraphics[width=8cm]{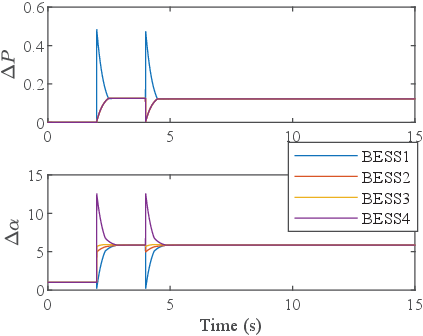}}
	\caption{Evolution of MC, power of each BESS, the total power, $\Delta P_i$, and $\Delta\alpha_i$ in Case 5\label{fig10}}
\end{figure}
\begin{figure}
	\centering
	\includegraphics[width=8cm]{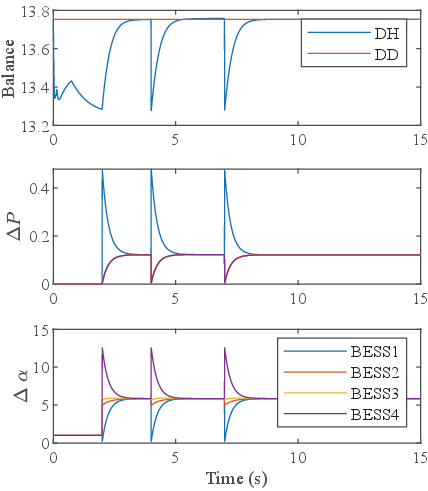}
	\caption{Evolution of the total power, $\Delta P_i$, and $\Delta\alpha_i$ in Case5\label{fig11}}
\end{figure}
\subsection{Case 6: Validation of the Plug and Play}
\par Considering the limitation of the State of Charge (SoC) level, when the SoC of a BESS is below a certain level, the output power drops to 0. In a sense, the capacity of this BESS decreases to 0, which implies this BESS exits. In this sense, plug and play scenarios are not only suitable for accidental and unplanned BESS exits, but also for dealing with BESS exits due to insufficient energy storage.
\par To verify the plug and play functionality, BESS 2 is removed at $t=10s$ and reconnected at $t=20s$ based on Case5. The simulation results are shown in Fig. \ref{fig12}.
\par From the simulation results, it can be seen that the designed solution can still achieve supply-demand balance, whether it is to connect or cut off a BESS. Meanwhile, the output power of each BESS does not exceed the limit. From another perspective, the designed solution is open, as it is also suitable for networks with more buses than BESS.
\begin{figure}
	\centering
	\subfigure[Evolution of MC, power of each BESS, and the total power\label{fig6a}]{\includegraphics[width=8cm]{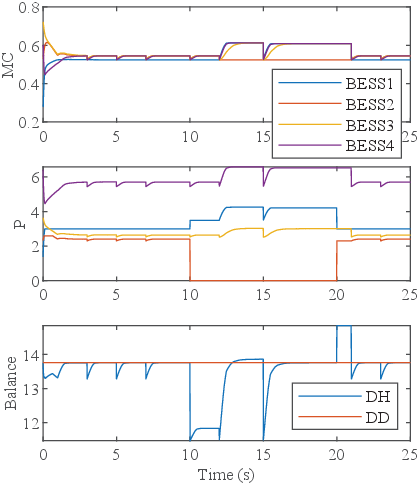}}
	\subfigure[Evolution of $\Delta P_i$ and $\Delta\alpha_i$\label{fig6b}]{\includegraphics[width=8cm]{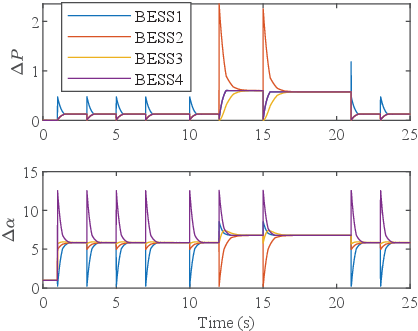}}
	\caption{Evolution of MC, power of each BESS, the total power, $\Delta P_i$, and $\Delta\alpha_i$ in Case 6\label{fig12}}
\end{figure}
\section{Conclusion}
In this article, we design a distributed ED scheme for a BESS network with PI+R controller to improve convergence speed and control accuracy while suppressing overshoot. Regularity, stability, and the enabling condition for the reset mechanism are derived. To ensure supply and demand balance, avoid Zeno behavior of reset mechanism, handle input delay, and adapt to limitations of capacity and SoC level, the above \textcolor{blue}{schemes} are modified by a centralized reset mechanism, a residence time, the reduced method, and a revised scheme, respectively. The simulation results show that the designed scheme is effective and efficient compared to existing results. However, the current reset mechanism is strictly centralized at the reset instants, and efforts may need to be made to design a more distributed reset mechanism in the future.
$\left[\begin{array}{ccc}
	1 & 2 & 3 
	\\
	4 & 3 & 4 
	\\
	3 & 5 & 3 
\end{array}\right]$
\appendices

\ifCLASSOPTIONcaptionsoff
\newpage
\fi
\small
\bibliographystyle{IEEEtran}
\bibliography{re1.bib}

@article{yong_capacity_2023,
	title = {Capacity Value of Uninterruptible Power Supply Storage},
	author = {Yong, Pei and Botterud, Audun and Zhang, Ning and Kang, Chongqing},
	year = {2023},
	month = mar,
	journal = {{IEEE} Transactions on Power Systems},
	volume = {38},
	number = {2},
	pages = {1763--1766},
	issn = {1558-0679},
	doi = {10.1109/TPWRS.2022.3233773},
}

@article{wangApplicationEnergyStorage2022,
	title = {Application of Energy Storage in Integrated Energy Systems \textemdash{} A Solution to Fluctuation and Uncertainty of Renewable Energy},
	author = {Wang, Wei and Yuan, Baoqiang and Sun, Qie and Wennersten, Ronald},
	year = {2022},
	month = aug,
	journal = {Journal of Energy Storage},
	volume = {52},
	pages = {104812},
	issn = {2352-152X},
	doi = {10.1016/j.est.2022.104812},
	urldate = {2023-03-13},
	langid = {english}
}

@article{solyaliComprehensiveStateoftheartReview2022,
	title = {A Comprehensive State-of-the-Art Review of Electrochemical Battery Storage Systems for Power Grids},
	author = {Solyali, Davut and Safaei, Babak and Zargar, Omid and Aytac, Gizem},
	year = {2022},
	month = oct,
	journal = {International Journal of Energy Research},
	volume = {46},
	number = {13},
	pages = {17786--17812},
	issn = {1099-114X},
	doi = {10.1002/er.8451},
	urldate = {2023-03-13},
	langid = {english}
}

@article{yangModellingOptimalEnergy2022,
	title = {Modelling and Optimal Energy Management for Battery Energy Storage Systems in Renewable Energy Systems: A Review},
	shorttitle = {Modelling and Optimal Energy Management for Battery Energy Storage Systems in Renewable Energy Systems},
	author = {Yang, Yuqing and Bremner, Stephen and Menictas, Chris and Kay, Merlinde},
	year = {2022},
	month = oct,
	journal = {Renewable and Sustainable Energy Reviews},
	volume = {167},
	pages = {112671},
	issn = {1364-0321},
	doi = {10.1016/j.rser.2022.112671},
	urldate = {2023-03-13},
	langid = {english}
}

@article{duanDistributedAlgorithmBased2021,
	title = {Distributed Algorithm Based on Consensus Control Strategy for Dynamic Economic Dispatch Problem},
	author = {Duan, Yuzhu and He, Xing and Zhao, You},
	year = {2021},
	month = jul,
	journal = {International Journal of Electrical Power \& Energy Systems},
	volume = {129},
	pages = {106833},
	issn = {0142-0615},
	doi = {10.1016/j.ijepes.2021.106833},
}

@article{ghanjatiOptimalSizingEnergy2022,
	title = {Optimal Sizing and Energy Management of a Stand-Alone Photovoltaic/Pumped Storage Hydropower/Battery Hybrid System Using Genetic Algorithm for Reducing Cost and Increasing Reliability},
	author = {Ghanjati, Chaima and Tnani, Slim},
	year = {2022},
	month = jul,
	journal = {Energy \& Environment},
	volume = {34},
	number = {6},
	pages = {2186-2203},
	publisher = {SAGE Publications Ltd STM},
	issn = {0958-305X},
	doi = {10.1177/0958305X221110529},
	urldate = {2023-03-13},
	langid = {english}
}

@article{rouholaminiReviewModelingManagement2022,
	title = {A Review of Modeling, Management, and Applications of Grid-Connected Li-Ion Battery Storage Systems},
	author = {Rouholamini, Mahdi and Wang, Caisheng and Nehrir, Hashem and Hu, Xiaosong and Hu, Zechun and Aki, Hirohisa and Zhao, Bo and Miao, Zhixin and Strunz, Kai},
	year = {2022},
	month = nov,
	journal = {IEEE Transactions on Smart Grid},
	volume = {13},
	number = {6},
	pages = {4505--4524},
	issn = {1949-3061},
	doi = {10.1109/TSG.2022.3188598}
}

@article{wangMultiagentbasedCollaborativeRegulation2022,
	title = {Multi-Agent-Based Collaborative Regulation Optimization for Microgrid Economic Dispatch under a Time-Based Price Mechanism},
	author = {Wang, Lingyun and An, Xiao and Xu, Honglei and Zhang, Yunning},
	year = {2022},
	month = dec,
	journal = {Electric Power Systems Research},
	volume = {213},
	pages = {108760},
	issn = {0378-7796},
	doi = {10.1016/j.epsr.2022.108760},
	urldate = {2023-02-10},
	langid = {english}
}

@article{yuFrequencySynchronizationPower2021b,
	title = {Frequency Synchronization and Power Optimization for Microgrids With Battery Energy Storage Systems},
	author = {Yu, Chang and Zhou, Hong and Lu, Xiaoqing and Lai, Jingang},
	year = {2021},
	month = sep,
	journal = {IEEE Transactions on Control Systems Technology},
	volume = {29},
	number = {5},
	pages = {2247--2254},
	publisher = {Ieee-Inst Electrical Electronics Engineers Inc},
	address = {Piscataway},
	issn = {1063-6536},
	doi = {10.1109/TCST.2020.3025606},
	urldate = {2022-10-01},
	langid = {english}
}

@article{tanExtensionsLocationalMarginal2022,
	title = {Extensions of the Locational Marginal Price Theory in Evolving Power Systems: A Review},
	shorttitle = {Extensions of the Locational Marginal Price Theory in Evolving Power Systems},
	author = {Tan, Zhenfei and Cheng, Tong and Liu, Yuchen and Zhong, Haiwang},
	year = {2022},
	month=apr,
	journal = {IET Generation, Transmission \& Distribution},
	volume = {16},
	number = {7},
	pages = {1277--1291},
	issn = {1751-8695},
	doi = {10.1049/gtd2.12381},
	urldate = {2023-03-13},
	langid = {english}
}

@article{wangDistributedOptimalPower2021,
	title = {Distributed Optimal Power Flow of DC Microgrids: A Penalty Based ADMM Approach},
	shorttitle = {Distributed Optimal Power Flow of DC Microgrids},
	author = {Wang, Manshang and Su, Yifan and Chen, Laijun and Li, Zhengming and Mei, Shengwei},
	year = {2021},
	month = mar,
	journal = {CSEE Journal of Power and Energy Systems},
	volume = {7},
	number = {2},
	pages = {339--347},
	publisher = {China Electric Power Research Inst},
	address = {Beijing},
	issn = {2096-0042},
	doi = {10.17775/CSEEJPES.2019.01650},
	urldate = {2023-03-15},
	langid = {english}
}

@article{liNewDistributedEnergy2021,
	title = {A New Distributed Energy Management Strategy for Smart Grid With Stochastic Wind Power},
	author = {Li, Wenjuan and Liu, Yungang and Liang, Huijun and Shen, Yanjun},
	year = {2021},
	month = feb,
	journal = {IEEE Transactions on Industrial Electronics},
	volume = {68},
	number = {2},
	pages = {1311--1321},
	issn = {1557-9948},
	doi = {10.1109/TIE.2020.2970627}
}

@article{hassanImprovedMantaRay2021,
	title = {An Improved Manta Ray Foraging Optimizer for Cost-Effective Emission Dispatch Problems},
	author = {Hassan, Mohamed H. and Houssein, Essam H. and Mahdy, Mohamed A. and Kamel, Salah},
	year = {2021},
	month = apr,
	journal = {Engineering Applications of Artificial Intelligence},
	volume = {100},
	pages = {104155},
	issn = {0952-1976},
	doi = {10.1016/j.engappai.2021.104155},
	urldate = {2023-03-14},
	langid = {english}
}

@article{ullahComputationallyEfficientConsensusBased2021,
	title = {A Computationally Efficient Consensus-Based Multiagent Distributed EMS for DC Microgrids},
	author = {Ullah, Md Habib and Babaiahgari, Bhanu and Alseyat, Anas and Park, Jae-Do},
	year = {2021},
	month = jun,
	journal = {IEEE Transactions on Industrial Electronics},
	volume = {68},
	number = {6},
	pages = {5425--5435},
	issn = {1557-9948},
	doi = {10.1109/TIE.2020.2992015}
}

@article{songCostBasedAdaptiveDroop2021,
	title = {Cost-Based Adaptive Droop Control Strategy for VSC-MTDC System},
	author = {Song, Sungyoon and McCann, Roy A. and Jang, Gilsoo},
	year = {2021},
	month = jan,
	journal = {IEEE Transactions on Power Systems},
	volume = {36},
	number = {1},
	pages = {659--669},
	issn = {1558-0679},
	doi = {10.1109/TPWRS.2020.3003589}
}

@article{sahooLocalizedEventDrivenResilient2021,
	title = {A Localized Event-Driven Resilient Mechanism for Cooperative Microgrid Against Data Integrity Attacks},
	author = {Sahoo, Subham and Peng, Jimmy Chih-Hsien},
	year = {2021},
	month = jul,
	journal = {IEEE Transactions on Cybernetics},
	volume = {51},
	number = {7},
	pages = {3687--3698},
	issn = {2168-2275},
	doi = {10.1109/TCYB.2020.2989225}
}

@article{zaeryNovelFullyDistributed2021,
	title = {A Novel Fully Distributed Fixed-Time Optimal Dispatch of DC Multi-Microgrids},
	author = {Zaery, Mohamed and Wang, Panbao and Wang, Wei and Xu, Dianguo},
	year = {2021},
	month = jul,
	journal = {International Journal of Electrical Power \& Energy Systems},
	volume = {129},
	pages = {106792},
	issn = {0142-0615},
	doi = {10.1016/j.ijepes.2021.106792},
	urldate = {2023-03-06},
	langid = {english}
}

@article{wangDisEHPPCEnablingHeterogeneous2022,
	title = {DisEHPPC: Enabling Heterogeneous Privacy-Preserving Consensus-Based Scheme for Economic Dispatch in Smart Grids},
	shorttitle = {DisEHPPC},
	author = {Wang, Aijuan and Liu, Wanping and Dong, Tao and Liao, Xiaofeng and Huang, Tingwen},
	year = {2022},
	month = jun,
	journal = {IEEE Transactions on Cybernetics},
	volume = {52},
	number = {6},
	pages = {5124--5135},
	issn = {2168-2275},
	doi = {10.1109/TCYB.2020.3027572}
}

@article{liDistributedControlStrategy2021,
	title = {Distributed Control Strategy for Global Economic Operation and Bus Restorations in a Hybrid AC/DC Microgrid with Interconnected Subgrids},
	author = {Li, Xiangke and Dong, Chaoyu and Jiang, Wentao and Wu, Xiaohua},
	year = {2021},
	month = oct,
	journal = {International Journal of Electrical Power \& Energy Systems},
	volume = {131},
	pages = {107032},
	issn = {0142-0615},
	doi = {10.1016/j.ijepes.2021.107032},
	urldate = {2023-03-06},
	langid = {english}
}

@article{pengDistributedPeriodicEventTriggered2022,
	title = {Distributed Periodic Event-Triggered Optimal Control of DC Microgrids Based on Virtual Incremental Cost},
	author = {Peng, Jiangkai and Fan, Bo and Tu, Zhenghong and Zhang, Wei and Liu, Wenxin},
	year = {2022},
	month = apr,
	journal = {IEEE/CAA Journal of Automatica Sinica},
	volume = {9},
	number = {4},
	pages = {624--634},
	issn = {2329-9274},
	doi = {10.1109/JAS.2022.105452}
}

@article{chenDistributedEconomicDispatch2021,
	title = {Distributed Economic Dispatch for Energy Internet Based on Multiagent Consensus Control},
	author = {Chen, Wushun and Li, Tao},
	year = {2021},
	month = jan,
	journal = {IEEE Transactions on Automatic Control},
	volume = {66},
	number = {1},
	pages = {137--152},
	issn = {1558-2523},
	doi = {10.1109/TAC.2020.2979749}
}

@book{spanglerPowerGenerationOperation2014,
	title = {Power Generation, Operation, and Control},
	author = {Robert and Shoults, Raymond},
	year = {2013},
	month = jul,
	publisher = {Wiley-Blackwell}
}

@article{zhaoDifferentialPrivacyEnergy2022,
	title = {Differential Privacy Energy Management for Islanded Microgrids With Distributed Consensus-Based ADMM Algorithm},
	author = {Zhao, Daduan and Zhang, Chenghui and Cao, Xiangyang and Peng, Chao and Sun, Bo and Li, Ke and Li, Yan},
	year = {2022},
	month = may,
	volume = {31},
	number = {3},
	journal = {IEEE Transactions on Control Systems Technology},
	pages = {1018--1031},
	issn = {1558-0865},
	doi = {10.1109/TCST.2022.3208456}
}

@article{chenDistributedCooperativeControl2021,
	title = {Distributed Cooperative Control for Economic Dispatch and SOC Balance in DC Microgrids with Vanadium Redox Batteries},
	author = {Chen, Shi and Gong, Qingwu and Lu, Xiaoqing and Lai, Jingang},
	year = {2021},
	month = dec,
	journal = {Sustainable Energy, Grids and Networks},
	volume = {28},
	pages = {100534},
	issn = {2352-4677},
	doi = {10.1016/j.segan.2021.100534},
	urldate = {2022-10-31},
	langid = {english}
}

@article{jinManageDistributedEnergy2022,
	title = {Manage Distributed Energy Storage Charging and Discharging Strategy: Models and Algorithms},
	shorttitle = {Manage Distributed Energy Storage Charging and Discharging Strategy},
	author = {Jin, Ruiyang and Lu, Chao and Song, Jie},
	year = {2022},
	month = jun,
	journal = {IEEE Transactions on Engineering Management},
	volume = {69},
	number = {3},
	pages = {755--764},
	issn = {1558-0040},
	doi = {10.1109/TEM.2020.3003306}
}

@article{hossainlipuReviewControllersOptimizations2022,
	title = {A Review of Controllers and Optimizations Based Scheduling Operation for Battery Energy Storage System towards Decarbonization in Microgrid: Challenges and Future Directions},
	shorttitle = {A Review of Controllers and Optimizations Based Scheduling Operation for Battery Energy Storage System towards Decarbonization in Microgrid},
	author = {Hossain Lipu, M. S. and Ansari, Shaheer and Miah, Md. Sazal and Hasan, Kamrul and Meraj, Sheikh T. and Faisal, M. and Jamal, Taskin and Ali, Sawal H. M. and Hussain, Aini and Muttaqi, Kashem M. and Hannan, M. A.},
	year = {2022},
	month = aug,
	journal = {Journal of Cleaner Production},
	volume = {360},
	pages = {132188},
	issn = {0959-6526},
	doi = {10.1016/j.jclepro.2022.132188},
	urldate = {2022-10-31},
	langid = {english}
}

@book{banosResetControlSystems2011,
	title = {Reset Control Systems},
	author = {Banos, Alfonso and Barreiro, Antonio},
	year = {2011},
	month = oct,
	publisher = {Springer London Ltd},
	urldate = {2022-12-01},
	isbn = {978-1-4471-2216-6}
}

@article{huResetControlConsensus2022,
	title = {Reset Control for Consensus of Double-Integrator Multi-Agent Systems},
	author = {Hu, Wenfeng and Cheng, Yi and Chen, Zhiyong},
	year = {2022},
	month = feb,
	journal = {Automatica},
	volume = {136},
	pages = {110057},
	issn = {0005-1098},
	doi = {10.1016/j.automatica.2021.110057},
	urldate = {2022-11-06},
	langid = {english}
}

@article{mengResetControlSynchronization2019,
	title = {Reset Control for Synchronization of Multi-Agent Systems},
	author = {Meng, Xiangyu and Xie, Lihua and Soh, Yeng Chai},
	year = {2019},
	month = jun,
	journal = {Automatica},
	volume = {104},
	pages = {189--195},
	issn = {0005-1098},
	doi = {10.1016/j.automatica.2019.02.042},
	urldate = {2022-11-06},
	langid = {english}
}

@ARTICLE{10302354,
	author={Zhang, Yalin and Liu, Zhongxin and Chen, Zengqiang},
	journal={IEEE Transactions on Smart Grid}, 
	title={A Marginal Cost Consensus Scheme With Reset Mechanism for Distributed Economic Dispatch in BESSs}, 
	year={2023},
	volume={},
	number={},
	pages={1-1},
	doi={10.1109/TSG.2023.3328937}}

@article{linNecessarySufficientGraphical2005,
	title = {Necessary and Sufficient Graphical Conditions for Formation Control of Unicycles},
	author = {Lin, Zhiyun and Francis, B. and Maggiore, M.},
	year = {2005},
	month = jan,
	journal = {IEEE Transactions on Automatic Control},
	volume = {50},
	number = {1},
	pages = {121--127},
	issn = {1558-2523},
	doi = {10.1109/TAC.2004.841121}
}

@article{cleggNonlinearIntegratorServomechanisms1958,
	title = {A Nonlinear Integrator for Servomechanisms},
	author = {Clegg, J. C.},
	year = {1958},
	month = mar,
	journal = {Transactions of the American Institute of Electrical Engineers, Part II: Applications and Industry},
	volume = {77},
	number = {1},
	pages = {41--42},
	issn = {2379-6774},
	doi = {10.1109/TAI.1958.6367399}
}

@article{zhangDistributedFiniteTimeMultiagent2019,
	title = {Distributed Finite-Time Multiagent Control for DC Microgrids With Time Delays},
	author = {Zhang, Runfan and Hredzak, Branislav},
	year = {2019},
	month=may,
	journal = {IEEE Transactions on Smart Grid},
	volume = {10},
	number = {3},
	pages = {2692--2701},
	issn = {1949-3061},
	doi = {10.1109/TSG.2018.2808467},
	eventtitle = {IEEE Transactions on Smart Grid}
}

@article{artsteinLinearSystemsDelayed1982,
	title = {Linear Systems with Delayed Controls: A Reduction},
	shorttitle = {Linear Systems with Delayed Controls},
	author = {Artstein, Z.},
	year = {1982},
	month=aug,
	journal = {IEEE Transactions on Automatic Control},
	volume = {27},
	number = {4},
	pages = {869--879},
	issn = {1558-2523},
	doi = {10.1109/TAC.1982.1103023},
	eventtitle = {IEEE Transactions on Automatic Control}
}

@article{Chen2018,
	title = {Delay effects on consensus-based distributed economic dispatch algorithm in microgrid},
	author = {Chen, Gang and Zhao, Zhongyuan},
	year = {2018},
	month=jan,
	journal = {IEEE Transactions on Power Systems},
	volume = {33},
	number = {1},
	pages = {602--612},
	issn = {1558-0679},
	doi = {10.1109/TPWRS.2017.2702179},
}
\end{document}